\documentclass[12pt]{article}

\usepackage{amsmath,amsthm,amssymb}
\usepackage[all]{xy}
\usepackage{graphicx}
\usepackage{verbatim}

\numberwithin{equation}{section}

\theoremstyle{plain}
\newtheorem{theorem}{\bf Theorem}[section]

\newtheorem{observation}[theorem]{\bf Observation}

\theoremstyle{definition}
\newtheorem{definition}{\bf Definition}[section]
\newtheorem{notation}[definition]{Notation}

\theoremstyle{remark}

\begin{document}

\title{On Level persistence (Relevant level persistence numbers)}

\author{Dong Du}

\maketitle


\begin{abstract}

The purpose of this note is to describe a new set of numerical invariants, \textit{the relevant level persistence numbers},
and make explicit their relationship with the four types of bar codes,  a more familiar set  of complete invariants for level persistence.
The paper provides the opportunity  to compare level persistence with the \textit{persistence} introduced by Edelsbrunner-Letscher-Zomorodian
 called in this paper as \textit{sub-level persistence}. 

\end{abstract}

\section {Introduction}

The \textit{level persistence} for a real valued map was first considered in \cite{DW} and \cite{BDD} and thought as a refinement of the standard persistence (referred below as \textit{sub-level persistence}). 
It turned out to be a particular case  of a more general persistence theory, the \textit{Zigzag persistence} proposed by Carlsson and Silva cf.\cite{CS}.
The numerical invariants we have proposed for level persistence  are the \textit{relevant persistence numbers} and are equivalent with the four types of bar codes which came out from \textit{Zigzag persistence}. Their merits consist in the fact that they can be calculated  using  standard persistence algorithms with minor adjustments. In \cite{D} we have indicated how to calculate these numbers  for a simplicial map via persistence algorithms slightly modified. The purpose of this note is to make this relationship precise.

In order to explain this we  review the meaning of level persistence versus sub-level persistence and explain, from our perspective, the significance of bar codes and of relevant persistence numbers. We inform the reader that the bar codes proposed by Carlsson and Silva are based on graph  representations and derived decomposing  the representations associated to the map in indecomposable components. Our approach is different.

We propose here two concepts \textit{death} (left and right) and \textit{observability or detectability} (left and right).
The class of maps for which level persistence is naturally defined based on these two concepts is the class of \textit{tame maps}. So far all  maps which appear in practice are tame. In particular any simplicial map $f: X\rightarrow \mathbb R$, where $X$ is a finite simplicial complex and $f$ is linear on each simplex, and any Morse function are tame. Tameness of a map actually signifies that the topology of the level changes at
a discrete collection of values (referred to as critical values).  Precisely, 
\begin{definition}

 \label{tamemap}

A continuous map $f:X\rightarrow  \mathbb{R}$  is called \textbf{tame map} (cf. Definition 3.5, \cite{BDD}) if
$X$ is a compact ANR and there exists finitely many values $\min(f(X)) = t_0<t_1<\cdots<t_N = \max(f(X))$ (so called critical values) so that

(i) for any $t\neq t_0,t_1,\cdots,t_N$ there exists $\epsilon>0$ so that $f:f^{-1}(t-\epsilon,t+\epsilon)\rightarrow(t-\epsilon,t+\epsilon)$
and the second factor projection $X_t\times (t-\epsilon,t+\epsilon)\rightarrow(t-\epsilon,t+\epsilon)$ are fiberwise homotopy equivalent.

(ii) for any $t_i$ there exists $\epsilon>0$ so that canonical inclusions $X_{t_i}\hookrightarrow  X_{t_i,t_i+\epsilon}$ and
$X_{t_i}\hookrightarrow  X_{t_i-\epsilon,t_i}$ are deformation retractions.

\end{definition}

The definition can be extended to incorporate $X$ locally compact ANR's  and $f$ proper maps.
Instead of finite collection of critical values one requires that the set of critical values  is a discrete sequence  of  numbers $\cdots t_i <t_{i+1}< t_{i+2} < \cdots $. 

One can show that  a simplicial map is tame cf.\cite{D}, with the set of critical values being among the values of $f$ on vertices. In practice, for a simplicial map,  one can treat all values of $f$ on  vertices  as potential critical values.

The sub-level persistence needs a weaker concept, referred here as weakly tame map, which requires the change in the topology of sub-levels appearing only at finitely many $t'$s.  Precisely,
\begin{definition}

 \label{wtamemap}

A continuous map $f:X\rightarrow  \mathbb{R}$  is called \textbf{weakly tame map}  if
$X$ is a compact ANR and there exist finitely many values $\min(f(X)) = t_0<t_1<\cdots<t_N = \max(f(X))$ (so called critical values) so that
for any $ t, t_i \leq t <t_{i+1}$ the inclusion $X_{(-\infty, t_i]}\subseteq X_{(-\infty, t]}$ is a homotopy equivalence (for the purpose of sub-level persistence,
homology equivalence suffices). As above the definition can be extended to locally compct ANR's.
\end{definition}

Clearly \textit{tameness} implies \textit{weakly tameness}. The main results stated here are Theorem \ref{theorem326} and Theorem \ref{theorem327}
 in section \ref{RAL} and they were formulated in the author's Ph.D thesis. 
As suggested, we begin this note with recollection of sub-level persistence (section \ref{SLP}),  then general considerations about level persistence (section \ref{LP}),  
and ultimately the relation between the relevant persistence numbers and bar codes (section \ref{RAL}). 

Note that when we refer to \textit{homology} we mean homology with coefficients in a field $\kappa$ fixed once for all. The case $\kappa= \mathbb Z_2$ and $\kappa= \mathbb R$ are the most familiar. In this case the $r$-dimensional homology is a $\kappa$-vector space and its dimension is referred to as Betti number.
The author thanks D. Burghelea for advise and help.


\section{Sublevel Persistence \cite{BD}} \label{SLP}

Given a continuous map $f:X\rightarrow \mathbb{R}$, the sub-level persistent homology introduced in \cite{ELZ} and further developed in \cite{ZC}
is concerned with the following questions:

Q1. Does the class $x\in H_r(X_{-\infty,t})$ originates in $H_r(X_{-\infty,t''})$ for $t''\leq t$ ? Does the class
$x\in H_r(X_{-\infty,t})$ vanishes in $H_r(X_{-\infty,t'})$ for $t<t'$ ?

Q2. What are the smallest $t'$ and $t''$ such that this happens?

The information that is contained in the linear maps $H_r(X_{-\infty,t})\rightarrow H_r(X_{-\infty,t'})$ for any $t\leq t'$
is known as \textbf{sub-level persistence} and permits to answer the above questions.

Recall that sub-level \textbf{persistent homology} is the collection of vector spaces and linear maps $\{ H_r(X_{-\infty,t})\rightarrow H_r(X_{-\infty,t'}), t<t', t,t'\in \mathbb R\}$.
\medskip 
Let $x\in H_r(X_{-\infty,t})), x\ne 0$. One says that

(i) The element $x\in H_r(X_{-\infty,t}))$ is born at $t''$, $t''\leq t$, if $x$ is contained in \newline img$(H_r(X_{-\infty,t''}))\rightarrow H_r(X_{-\infty,t})))$
but is not contained in img$(H_r(X_{-\infty,t''-\epsilon}))\rightarrow H_r(X_{-\infty,t})))$ \newline for any $\epsilon>0$.

(ii) The element $x\in H_r(X_{-\infty,t}))$ dies at $t'$, $t'>t$, if its image is zero in \newline img$(H_r(X_{-\infty,t}))\rightarrow H_r(X_{-\infty,t'})))$
but is nonzero in  img$(H_r(X_{-\infty,t}))\rightarrow H_r(X_{-\infty,t'-\epsilon})))$ \newline for any $0<\epsilon<t'-t$.

(iii) The element $x\in H_r(X_{-\infty,t}))$   survives for ever, if its image is always nonzero in \newline img$(H_r(X_{-\infty,t}))\rightarrow H_r(X_{-\infty,t'})))$ for any $t'>t$.

\vskip 3mm

Note that most papers treat persistence for filtered spaces rather than for a map. Clearly a map provide a filtration by finitely many sub-levels if the map is weakly tame. 
Conversely, the standard construction \textit{telescope} in homotopy theory
permits to replace any finite filtered space
$K_0\subseteq K_1\subseteq \cdots\subseteq K_N$
by a weakly tame map  $f:X\rightarrow \mathbb{R}$  simply by taking
$X = K_0\times [t_0,t_1] \cup_{\phi_1} K_1\times [t_1,t_2] \cup_{\phi_2}\cdots\cup_{\phi_{N-1}}K_{N-1}\times [t_{N-1},t_N]\cup_{\phi_N}K_N$
where $\phi_i: K_i\times \{t_{i+1}\}\rightarrow K_{i+1}\times \{t_{i+1}\}$ is the inclusion
and $f|_{K_i\times [t_i,t_{i+1}]}$  the projection of
$K_i\times [t_i,t_{i+1}]$ on $[t_i,t_{i+1}]$.

The sub-level persistence  for a filtered space is the sub-level persistence of the associated weakly tame map.

When $f$ is weakly tame, the sub-level persistence for each $r = 0, 1, \cdots, \dim X$ is determined by a finite collection of invariants
referred to as \textbf{bar codes for sub-level persistence} \cite{ZC}.
The $r$-bar codes for sub-level persistence of $f$ are  intervals of the form $[t, t')$ or $[t, \infty)$
with $t< t'$.

The number $\mu_r(t,t')$ of  $r$-bar codes which identify to the interval $[t, t')$ is the maximal number of linearly independent
homology classes in $H_r(X_{-\infty,t})$, which are born at $t$, die at $t'$ and remain independent in
img($H_r(X_{-\infty,t})\rightarrow H_r(X_{-\infty,s})$) for any $s$, $t\leq s<t'$.

The number $\mu_r(t,\infty)$ of $r$-bar codes which identify to the interval $[t, \infty)$ is the maximal number of linearly independent
homology classes  in $H_r(X_{-\infty,t})$ which are born at $t$, 
and remain independent in
img($H_r(X_{-\infty,t})\rightarrow H_r(X_{-\infty,s})$) for any $s>t$.

It follows from the above definitions  that for a weakly tame map 
the set of $r$-bar codes for sub-level persistence is  finite and any $r$-bar code  is an interval of the form $[t_i, t_j)$ or $[t_i, \infty)$
with $t_i, t_j$  critical values of $f$ and $t_i< t_j$.

From these bar codes one can derive the \textbf{Betti numbers} $\beta_r(t, t')$, the dimension of
\newline $\text{img}(H_r(X_{-\infty,t})\rightarrow H_r(X_{-\infty,t'}))$, for any $t\leq t'$ and get the answers to questions Q1 and Q2.

For example,
\begin{equation}\label {E1}
\beta_r(t, t') = \text{ the number of } r \text{-bar codes which contain the interval } [t, t'].
\end{equation}

From the Betti numbers $\beta_r(t,t')$ one can also derive these $r$-bar codes.

Denote $\mu_r(t_i,t_j)$=number of $r$-bar codes
which equal to $[t_i,t_j)$ for  $t_0\leq t_i< t_j\leq \infty$, where
$t_0$ is the smallest critical value. We have (see \cite{EH})
\begin{equation}\label{E2}
\begin{array}{ll}
& \mu_r(t_i,t_j) \\ [3mm]
= & \hskip -1.5mm
 \left\{ \hskip -1.5mm
     \begin{array}{ll}
       \beta_r(t_i,t_{j-1})-\beta_r(t_{i-1},t_{j-1})-\beta_r(t_i,t_j)+\beta_r(t_{i-1},t_j), & t_0<t_i< t_j<\infty\\
       \beta_r(t_0,t_{j-1})-\beta_r(t_0,t_j), & t_i=t_0, t_0< t_j<\infty\\
       \beta_r(t_i,\infty)-\beta_r(t_{i-1},\infty), & t_0<t_i<\infty,t_j=\infty\\
       \beta_r(t_0,\infty), & t_i=t_0,t_j=\infty
     \end{array}
   \right.
\end{array}
\end{equation}

The computation of the bar codes for a filtration of simplicial or polytopal
complex or equivalently for a simplicial map is discussed in {subsection 3.4 of \cite{D}}
when the coefficients field  for homology groups is  $\mathbb{Z}_2$ or $\mathbb{R}$. The case of the field $\kappa=\mathbb Z_2$ is taken from \cite {ELZ}
and is, by now, the well known ELZ-algorithm .

\section{Level Persistence\cite{BD}} \label{LP}

Level persistence for a map $f:X\to \mathbb R$ was first considered in \cite{DW} and was better understood when the Zigzag persistence
was introduced and formulated in \cite{CSM}.
Given a continuous map $f:X\rightarrow \mathbb{R}$,
level persistence is concerned with the homology of the fibers $H_r(X_t)$ and addresses
questions of the following type.

Q1. Does the image of $x\in H_r(X_t)$ vanish in $H_r(X_{t, t'})$, where $t'>t$ or in $H_r(X_{t'',t})$, where $t''<t$ ?

Q2. Can $x\in H_r(X_t)$ be detected in $H_r(X_{t'})$ where $t'>t$ or in $H_r(X_{t''})$ where $t''<t$ ? The precise meaning of
detection is explained below.

Q3. What are the smallest $t'$ and $t''$ for the answers to Q1 and Q2 to be affirmative?

\vskip 3mm

To answer such questions one has to record information about the following linear maps
$$
H_r(X_t)\rightarrow H_r(X_{t, t'})\leftarrow H_r(X_{t'}).
$$
The \textbf{level persistence} is the information provided by this collection of vector spaces and linear maps  considered
for all $t$, $t'$.

Let $0\neq c\in H_r(X_t)$. One says that

(i) $c$ dies downward at $t'< t$, if its image is zero in img$(H_r(X_t)\rightarrow H_r(X_{t', t}))$
but is nonzero in img$(H_r(X_t)\rightarrow H_r(X_{t'+\epsilon,t}))$ for any $0<\epsilon<t-t'$.

(ii) $c$ dies upward at $t''>t$, if its image is zero in img$(H_r(X_t)\rightarrow H_r(X_{t, t''}))$
but is nonzero in img$(H_r(X_t)\rightarrow H_r(X_{t, t''-\epsilon}))$ for any $0<\epsilon<t''-t$.

We say that $x\in H_r(X_t)$ can be detected at $t'\geq t$, if its image in $H_r(X_{t,t'})$ is nonzero and is
contained in the image of $H_r(X_{t'})\rightarrow H_r(X_{t,t'})$.
Similarly, the detection of $x$ can be defined for $t''<t$ also.

In case of sub-level persistence for  tame maps the collection of the vector spaces and linear maps is determined up to coherent isomorphisms
by a collection of invariants called \textbf{bar codes for level persistence} which are intervals of the form
$[t, t']$ with $t\leq t'$ and  $(t, t')$, $(t, t']$, $[t, t')$ with $t< t'$.

These bar codes are called {\it invariants}  because two tame maps $f: X\rightarrow \mathbb{R}$ and $g: Y\rightarrow \mathbb{R}$
which are fiber-wise homotopy equivalent have the same associated bar codes.
The above result can be derived  from  Zigzag persistence  but, in view of definitions above can be proven directly. The details of the derivation are not contained 
 in this paper.

An open end of an interval signifies the death of a homology class at that end (left or right) whereas a closed end signifies
that a homology class cannot be detected beyond this level (left or right).

There exists an $r$-bar code $(t'', t')$ if there exists a class $x\in H_r(X_{t})$ for some $t''<t<t'$
which is detectable for $t''< s< t'$ and dies at $t''$ and $t'$.
The multiplicity of $(t'', t')$ is the maximal number of linearly independent classes
in $H_r(X_{t})$ such that

(i) all remain linearly independent in img($H_r(X_{t})\rightarrow H_r(X_{t,s})$) for $t\leq s< t'$
and  \newline img($H_r(X_{t})\rightarrow H_r(X_{s,t})$) for $t''< s\leq  t$;

(ii) all die at $t''$ and $t'$.

\vskip 3mm

Notice that the change of $t$ above will not affect the multiplicity of $(t'',t')$.

There exists an $r$-bar code $(t'', t']$ if there exists an element $x\in H_r(X_{t'})$
which is not detectable for $s>t'$ and detectable for $t''< s\leq t'$ and dies at $t''$.
The multiplicity of $(t'', t']$ is the maximal number of linearly independent elements
in $H_r(X_{t'})$ such that

(i) neither one is detectable for $s>t'$;

(ii) all remain linearly independent in img($H_r(X_{t'})\rightarrow H_r(X_{s,t'})$) for $t''< s\leq t'$;

(iii) all dies at $t''$.

\vskip 3mm

There exists an $r$-bar code $[t'', t')$ if there exists an element $x\in H_r(X_{t''})$
which is not detectable for $s<t''$ and detectable for $t''\leq s<t'$ and dies at $t'$.
The multiplicity of $[t'', t')$ is the maximal number of linearly independent elements
in $H_r(X_{t''})$ such that

(i) neither one is detectable for $s<t''$;

(ii) all remain linearly independent in img($H_r(X_{t''})\rightarrow H_r(X_{t'',s})$) for $t''\leq s<t'$;

(iii) all dies at $t'$.

\vskip 3mm

There exists an $r$-bar code $[t'', t']$ if there exists an element $x\in H_r(X_{t''})$
which is not detectable for $s<t''$ or $s>t'$ and detectable for $t''\leq s\leq t'$.
The multiplicity of $[t'', t']$ is the maximal number of linearly independent elements
in $H_r(X_{t''})$ such that

(i) neither one is detectable for $s<t''$ or $s>t'$;

(ii) all remain linearly independent in img($H_r(X_{t''})\rightarrow H_r(X_{t'',s})$) for $t''\leq s\leq t'$.

\vskip 3mm

Note, that a priory, the set of linearly independent elements in $H_r(X_t)$ for each $t$ between $t'$ and $t''$ might be very different for different  $t'$s.  The tameness hypothesis insures however their consistency.

In view of the description  above 
for a tame map, the set of $r$-bar codes for level persistence is finite. Any  $r$-bar code  is an  interval of the form
$[t_i, t_j]$ with $t_i\leq t_j$ critical values or   $(t_i, t_j)$, $(t_i, t_j]$, $[t_i, t_j)$ with $t_i< t_j,$ $t_i, t_j$ critical values.

%
%

\begin{notation}
Given a tame map $f:X\rightarrow \mathbb{R}$ with critical values $t_0<\cdots<t_N$,
denote by
$$
\begin{array}{l}
\,\, BL_r(f) := \text{the number of all $r$-bar codes for level persistence} \\
                     \text{\hskip 1.9cm (with respect to \emph{r-}th homology groups)}.\\
N_r(t_i,t_j) := \text{the number of intervals } (t_i, t_j) \text{ in } BL_r(f). \\
N_r(t_i,t_j] := \text{the number of intervals } (t_i, t_j] \text{ in } BL_r(f). \\
N_r[t_i,t_j) := \text{the number of intervals } [t_i, t_j) \text{ in } BL_r(f). \\
N_r[t_i,t_j] := \text{the number of intervals } [t_i, t_j] \text{ in } BL_r(f). \\
\end{array}
$$
Hence $\sharp BL_r(f) = N_r(t_i,t_j) + N_r(t_i,t_j] + N_r[t_i,t_j) + N_r[t_i,t_j]$.

\end{notation}

In {Figure \ref{figure421}}, we indicate the bar codes both for sub-level and level persistence for some simple map in order to illustrate their
differences and what they have in common.  

For example looking at {Figure \ref{figure421}} the class consisting of the sum of two circles at level $t$ is not detected on the right,
but is detected at all levels on the left up to (but not including) the level $t'$.

Level persistence provides considerably more information than the sub-level persistence \cite{BDD}
and the bar codes for the sub-level persistence can be recovered from the bar codes for the level
persistence.
An $r$-bar code $[s, t)$ for level persistence contributes an $r$-bar code $[s,t)$ for sub-level persistence.
An $r$-bar code $[s, t]$ for level persistence contributes an $r$-bar code $[s,\infty)$ for sub-level persistence.
$r$-bar codes $(s,t]$ and $(s,t)$ for level persistence contribute nothing to $r$-bar codes for sub-level persistence.
An $r$-bar codes $(s,t)$  for level persistence contributes an $r+1$-bar code $[t,\infty)$  for sub-level persistence.
See {Figure \ref{figure421}} and {Lemma \ref{levsublev}} below.

\begin{figure}[h!]

\includegraphics[scale=.4]{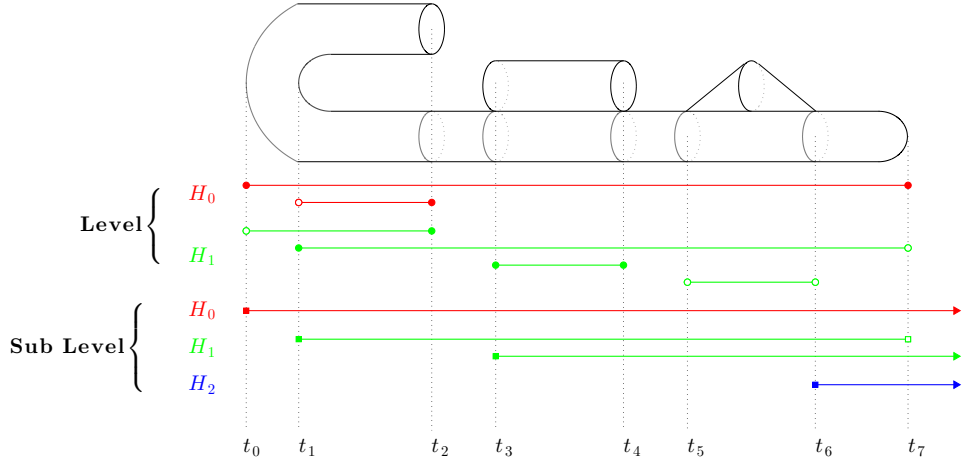}

\caption{Bar codes for level and sub-level persistence.} \label {figure421}

\end{figure}

\begin{theorem} \label{levsublev}
Given a tame map $f: X\rightarrow \mathbb{R}$ with critical values $t_0<t_1<\cdots<t_N$.
We have
$$
\mu_r(t_i,t_j) = N_r[t_i, t_j)
$$
$$
\mu_r(t_i,\infty) = \sum_{l=i}^N N_r[t_i, t_l] + \sum_{l=0}^{i-1} N_{r-1}(t_l,t_i)
$$
for any critical values $t_i<t_j$.
\end{theorem}

\begin{proof} 
Item 1 follows from formulas (\ref{E1}) and (\ref{E2}). 

Item 2 is more elaborate. One uses formula (\ref{E2}) which calculates $\mu_r(t_i, \infty)$ as $\mu_r(t_i, \infty)= \beta_r(t_i, \infty) - \beta_r(t_{i-1}, \infty)$.  A calculation of $\beta_r(t_i, \infty)$ can be recovered from Corollary 3.4 in \cite{BD} which implies that 
this number is exactly the number of $(r-1)$-bar codes of the form $(t_l, t_i), l= 0,1, \cdots, i-1$ plus the number of $r$-bar codes of the form $[a,b]$ with $a\leq t_i$.
Clearly $a, b$ should be critical values. A different derivation can be achieve independently of \cite {BD}.
\end{proof}

The bar codes for the level persistence can be also recovered from the bar codes for the sub-level persistence but from the bar codes of a collections of tame maps canonically associated to $f$.
This will be described in the next subsection.

\section{Relations Between Relevant Persistence Numbers and Bar Codes} \label{RAL}

For this purpose one uses an alternative but equivalent way to describe the level persistence based on a different collection of numbers,  referred below as relevant persistence numbers, $l_r, l_r^+, l^-_r, e_r, i_r$. 

\begin{definition}
For a continuous map $f: X\rightarrow \mathbb{R}$ and $t''\leq t\leq t'$,
let $L_r(t):=H_r(X_t)$, $L_r^+(t;t'):=\ker(H_r(X_t)\rightarrow H_r(X_{t, t'}))$,
$L_r^-(t;t''):=\ker(H_r(X_t)\rightarrow H_r(X_{t'',t}))$
and $$I_r(t, t'):= \text{img}(H_r(X_t)\rightarrow H_r(X_{t,t'})) \cap \text{img}(H_r(X_{t'})\rightarrow H_r(X_{t,t'})).$$

Define the \textbf{relevant level persistent numbers}

\begin{enumerate}
\item
$l_r(t):= \dim L_r(t)$
\item
$l_r^+(t;t'):=\dim L_r^+(t;t')$
\item
$l_r^-(t;t''):=\dim L_r^-(t;t'')$
\item  $e_r(t;t',t''):=\dim(L_r^+(t;t')\cap L_r^-(t;t''))$ 
\newline \newline and
\item $i_r(t,t'):=\dim(I_r(t,t'))$
\end{enumerate}

\end{definition}

The relation between these collections of numbers is illustrated in the diagram below. 

$$
\xymatrix
{
	*+[F]\txt{$i_r(t, t')$} \ar@{=>}[rr]^{\hskip -1cm Thm\; 4.2} &
	 &
	 *+[F]\txt{$ l_r(t), l_r^+(t;t')$\\$l_r^-(t;t''), e_r(t;t',t'')$} \ar@{=>}[rr]^{\hskip 0.5cm Thm\; 4.3} &
	 &
	*+[F]\txt{$N_r([t, t'])$\\ $N_r((t, t'))$\\ $N_r((t, t'])$\\ $N_r([t, t'))$} \ar@/_-4pc/[ll]^{\hskip -.8cm Observation\; 4.1} \ar@/_-8pc/[llll]^{Observation\; 4.1}
}
$$

The first four have geometric meaning the last ones (the fifth) are more technical. However the first four $l_r, l_r^+, l^-_r, e_r$ can be derived from the last ones, $i_r$
by Theorem \ref{theorem326} 

One can derive all the numbers $l_r, l_r^+, l^-_r, e_r, $ as well as $i_r$  from the number of bar codes $N_r(t_i, t_j)$, $N_r(t_i, t_j]$, $N_r[t_i, t_j)$, $N_r[t_i, t_j]$
by Observation \ref{O1}.

\begin{observation}\label {O1}
For a tame map
we can derive relevant level persistent numbers
from the numbers $N's$ of bar codes for level persistence.
\end{observation}

\begin{proof}

For $t''\leq t\leq t'$
\begin{enumerate}
\item
$i_r(t,t')$ = number of intervals in $BL_r(f)$ which contain $[t, t']$;
\item
$l_r(t)$ = number of intervals in $BL_r(f)$ which contain $t$;
\item 
\begin{equation} \label{l+}
l_r^+(t; t') = \sum_{t_i\leq t< t_j\leq t'} N_r[t_i,t_j) + \sum_{t_i< t<t_j\leq t'} N_r(t_i,t_j);
\end{equation}
\item 
\begin{equation} \label{l-}
l_r^-(t; t'') = \sum_{t''\leq t_i< t\leq t_j} N_r(t_i,t_j] + \sum_{t''\leq t_i< t< t_j}  N_r(t_i,t_j);
\end{equation}
\item
\begin{equation} \label{e_r}
e_r(t;t',t'') = \sum_{ t''\leq t_i<t<t_j\leq t'} N_r(t_i,t_j).
\end{equation}
\end{enumerate}
\end{proof}

 
\begin{theorem} \label{theorem326}

For a tame map the numbers $i_r(t,t')$
determine the numbers  $l_r(t)$, $l_r^+(t;t')$, $l_r^-(t;t'')$
 and  $e_r(t;t',t'')$. 
\end{theorem}

\begin{proof}

$$
l_r(t) = i_r(t,t).
$$

$$
N_r(t_k, t_j) = i_r(t'',t') - i_r(t_k, t') - i_r(t'', t_j) + i_r(t_k, t_j)
$$
for any $t''\leq t'$ such that $t_k<t''<t_{k+1}$, $t_{j-1}<t'<t_j$.

$$
N_r[t_k, t_j] = i_r(t_k, t_j) - i_r(t'', t_j) - i_r(t_k, t') + i_r(t'', t')
$$
for any $t_{k-1}<t''<t_k$ and $t_j<t'<t_{j+1}$.

$$
N_r(t_k,t_j] = i_r(t'',t_j) - i_r(t_k,t_j) - i_r(t'',t') + i_r(t_k,t')
$$
for any $t_k<t''<t_{k+1}$ and $t_j<t'<t_{j+1}$.

$$
N_r[t_k,t_j) = i_r(t_k,t') - i_r(t_k,t_j) - i_r(t'',t') + i_r(t'',t_j)
$$
for any $t_{k-1}<t''<t_k$ and $t_{j-1}<t'<t_j$.

Plug in equation (\ref{l+}), (\ref{l-}) and (\ref{e_r})
we get $l_r^+(t;t')$, $l_r^-(t;t'')$ and  $e_r(t;t',t'')$.

\end{proof}

\begin{theorem} \label{theorem327}
For a tame map
the relevant persistent numbers $\{ l_r, l_r ^+, l_r^-, e_r\}$ determine the bar codes $N_r(\cdots)'$s.
\end{theorem}

\begin{proof}
First the numbers $ N_r(t_k, t_j)$ can be calculated by the formula.

\begin{equation}
N_r(t_k, t_j) = e_r(t; t_j,t_k) - e_r(t; t_j,t_{k+1})
- e_r(t; t_{j-1},t_k) + e_r(t; t_{j-1},t_{k+1}),
\end{equation}
for any $t_k<t<t_j$.

To determine the numbers $N_r(t_i,t_j]$, $N_r[t_i,t_j)$ and $N_r[t_i,t_j]$,
we introduce the following auxiliary numbers 
$$
\begin{array}{l}
n_r\{t_i,t_j\} := \text{the number of intervals in } BL_r(f) \text{ which intersect the levels }\\
                       \hskip 2.1cm X_{t_i} \text{ and } X_{t_j};\\
n_r\{t_i,t_j) := \text{the number of intervals in } BL_r(f) \text{ which intersect the level } X_{t_i} \\
                      \hskip 2.1cm \text{ with open end at } t_j;\\
n_r\{t_i,t_j] := \text{ the number of intervals in } BL_r(f) \text{ which intersect the level } X_{t_i} \\
                      \hskip 2.1cm \text{ with closed end at } t_j;\\
n_r(t_i,t_j\} := \text{ the number of intervals in } BL_r(f) \text{ which intersect the level } X_{t_j} \\
                      \hskip 2.1cm \text{ with open end at } t_i;\\
n_r[t_i,t_j\} := \text{ the number of intervals in } BL_r(f) \text{ which intersect the level } X_{t_j} \\
                      \hskip 2.1cm \text{ with closed end at } t_i.
\end{array}
$$

The numbers $ n_r\{t_i,t_j),$  $n_r\{t_i,t_j),$   $n_r\{t_i,t_j\}$ and $n_r[t_i,t_j\}$ can be derived from the
 relevant persistent numbers as indicated below

$$
\begin{array}{rcl}
n_r\{t_i,t_j) &=& l_r^+(t_i;t_j) - l_r^+(t_i;t_{j-1}) \\
n_r(t_i,t_j\} &=& l_r^-(t_j;t_i) - l_r^-(t_j;t_{i+1}) \\
n_r\{t_i, t_j\} &=& i_r(t_i,t_j) \\
n_r[t_i,t_j\} &=& n_r\{t_i,t_j\} - n_r\{t_{i-1},t_j\} - n_r(t_{i-1},t_j\}
\end{array}
$$

With their help one derive 

\begin{equation}
N_r(t_i,t_j] = n_r(t_i,t_j\} - n_r(t_i,t_{j+1}\} - N_r(t_i,t_{j+1})
\end{equation}
\begin{equation}
N_r[t_i,t_j) = n_r\{t_i,t_j) - n_r\{t_{i-1},t_j) - N_r(t_{i-1},t_j)
\end{equation}
\begin{equation}
N_r[t_i,t_j] = n_r[t_i,t_j\} - n_r[t_i,t_{j+1}\} - N_r[t_i,t_{j+1})
\end{equation}

\end{proof}

\vskip .5cm

The explicit calculation of  the 
relevant persistence numbers $l_r(t)$, $l_r^+(t;t')$, $l_r^-(t;t'')$
and $e_r(t;t',t'')$ is discussed  in {subsection 4.4 of \cite{D}} and is based on positive and negative bar codes  which are defined and calculated in terms of sub-level persistence via minor adjustments of the ELZ algorithm.

%

Alternatively,  we can  get the bar codes for the level persistence providing an alternative to the Carson-Silva algorithm cf. \cite{CS}  which calculates the level persistence bar codes as bar codes for Zigzag persistence.



\newpage

\begin{thebibliography}{99}

%

\bibitem{BD} D. Burghelea, T. K. Dey. Defining and Computing Topological Persistence for 1-cocycles.
arXiv:1104.5646v3, 2011.

\bibitem{BDD} D. Burghelea, T. K. Dey. Topological Persistence for Circle Valued Maps.
Discrete Comput. Geom. 50 (2013), no. 1, 69-98.

%
%
%
%
%

\bibitem{CS} G. Carlsson and V. D. Silva. Zigzag Persistence. \textit{Foundations of Computational Mathematics}, 10(4): 367-405, 2010.


\bibitem{CSM} G. Carlsson, V. D. Silva and D. Morozov. Zigzag Persistent Homology and Real-valued Functions. \textit{Proc. 25th Annu. Sympos. Comput. Geom.}, 247-256, 2009.

%
%
%

\bibitem{DW} T. K. Dey and R. Wenger. Stability of Critical Points with Interval Persistence. \textit{Discrete Comput. Geom.}, 38: 479-512, 2007.

\bibitem{D} D. Du. Contributions to Persistence Theory, arXiv:1210.3092v3 [cs.CG], 2014.

%


\bibitem{EH}  H. Edelsbrunner, J. L. Harer. Computational Topology: An Introduction, AMS Press, 2010.

\bibitem{ELZ}  H. Edelsbrunner, D. Letscher, and A. Zomorodian. Topological persistence and simplification. \textit{Discrete Comput. Geom.}, 28: 511-533, 2002.

%
%
%
%
%
%
%
%
%
%
%
%
%
%
%
%
%
%

\bibitem{ZC}  A. J. Zomorodian and G. Carlsson. Computing Persistent Homology. \textit{Discrete Comput. Geom.}, 33: 249-274, 2005.




\end{thebibliography}
\end{document}